\documentclass[preprint]{elsarticle}

\usepackage[utf8x]{inputenc}
\usepackage[american]{babel}

\usepackage{hyperref}
\usepackage{amssymb}
\usepackage{amsmath}
\usepackage{amsthm}
\usepackage{xspace}
\usepackage{todonotes}
\usepackage{amssymb}

\usepackage{framed}

\journal{Journal of Computer and System Sciences}

\newtheorem*{mtheorem}{Main Theorem}
\newtheorem*{mcorollary}{Main Corollary}

\newtheorem{theorem}{Theorem}
\newtheorem{lemma}[theorem]{Lemma}
\newtheorem{corollary}[theorem]{Corollary}
\newtheorem{conjecture}[theorem]{Conjecture}

\newtheorem{claim}{Claim}
\newtheorem{fact}[theorem]{Fact}

\newcommand{\calH}{\mathcal{H}}
\newcommand{\cH}{\mathcal{H}}
\newcommand{\calC}{\mathcal{C}}
\newcommand{\tr}{\mathsf{tr}}
\newcommand{\mtr}{\mathsf{mtr}}
\newcommand{\btr}{\mathsf{btr}}

\newcommand{\calB}{B}
\newcommand{\calA}{\mathcal{A}}
\newcommand{\cN}{\mathcal{N}}

\newcommand{\xcup}{\bigotimes}
\renewcommand{\emptyset}{\varnothing}
\renewcommand{\sharp}{\#}

\def\wrt{\emph{w.r.t.}\xspace}
\def\ie{\emph{i.e.}\xspace}

\begin{document}

\begin{frontmatter}
	
\title{Counting Minimal Transversals of $\beta$-Acyclic Hypergraphs\tnoteref{label1}}

\author[label2]{Benjamin Bergougnoux}
\ead{benjamin.bergougnoux@uca.fr}
\author[label3]{Florent Capelli}
\ead{florent.capelli@univ-lille.fr}
\author[label2]{Mamadou Moustapha Kant\'e}
\ead{mamadou.kante@uca.fr}

\address[label2]{Université Clermont Auvergne, LIMOS, CNRS, France} 

\address[label3]{Université de Lille, CRIStAL, CNRS, INRIA, France}

\tnotetext[label1]{The authors  are supported by the French Agency for Research under the GraphEN (ANR-15-CE40-0009) and AGGREG (ANR-14-CE25-0017)
  projects.}

\begin{abstract} We prove that one can count in polynomial time the number of minimal transversals of $\beta$-acyclic hypergraphs. In consequence, we can count in polynomial time the number of minimal dominating sets of strongly chordal graphs, continuing the line of research initiated in [M.M. Kanté and T. Uno, Counting Minimal Dominating Sets, TAMC'17]. 
\end{abstract}


\end{frontmatter}

\section*{Introduction}\label{sec:intro}

A \emph{hypergraph} is a collection of subsets - called \emph{hyperedges} - of a finite ground set, and a \emph{transversal} is a subset of the ground set that intersects
every hyperedge.  In this paper, we consider the problem of counting the (inclusion-wise) minimal transversals of \emph{$\beta$-acyclic hypergraphs}. Counting problems
are usually harder than decision problems as for instance computing a (minimal) transversal of a hypergraph can be done in polynomial time while counting the number of
(minimal) transversals is $\#$P-complete \cite{Vadhan01}.  Informally, $\#$P denotes the set of functions corresponding to the number of accepting paths of a
non-deterministic Turing machine, and FP $\subseteq$ $\#$P, the set of functions computable in polynomial time. Under the assumption $\#$P $\ne$ FP, one may be interested
in classifying counting problems between those that are easy to compute, \ie, belong to FP, and those that are hard, \ie, are $\#$P-hard \cite{Valiant79}, or even hard to
approximate \cite{GoldbergGL16}.

The \textsc{Hypergraph Dualisation} problem -- a fifty years open problem -- asks for the enumeration of all (inclusion-wise) minimal transversals of a given hypergraph. It is not yet known whether this problem can be solved by an algorithm that runs in time polynomial in the size of the output despite the fact that it has been extensively studied since it has several applications in many areas such as graph theory, artificial intelligence, datamining, model-checking, network modeling, or databases (see the survey \cite{EiterGM08}). For all these applications, the minimality of the transversals is needed. Thus, even if it is not hard to see that listing all transversals of a hypergraph can be done with a polynomial delay between two solutions, there may be an exponential gap between the number of minimal transversals and the number of transversals (see for instance the non-trivial class of affine formulas), which makes efficient enumeration of minimal transversals much harder.

The problem of computing the number of minimal transversals of a hypergraph is closely related to the \textsc{Hypergraph Dualisation} problem and has also many applications in several areas, see for instance the description given in \cite{DurandH08} in the case of model checking. This problem has also applications in graph theory as it is closely related to the problem of counting the \emph{minimal dominating sets} of a graph. It turns out that these counting problems are $\#$P-complete in general, that is, it is very unlikely that it can be solved in polynomial time on every input.  However, there is a rich literature on solving this problem in polynomial time by restricting the input hypergraph to a specific class, see for example~\cite{CapelliPhd2016,KanteLMNU13,GolovachHKKSV17}.

 \paragraph{Our contributions and approach} In this paper, we follow this line of research by exhibiting a new tractable class of hypergraphs for this problem, namely, the class of $\beta$-acyclic hypergraphs. 
 The \textsc{Hypergraph Dualisation} problem was already known to be tractable for this class~\cite{EiterG2002}. 
 Moreover, it has been proved in \cite{Brault-BaronCM15} that the number of transversals of a $\beta$-acyclic hypergraph can be computed in polynomial time but the complexity of computing the number of minimal transversals was still unknown. More precisely, we prove the following theorem:
 
\begin{mtheorem}
  The number of minimal transversals of a $\beta$-acyclic hypergraph can be computed in polynomial time.
\end{mtheorem}

A direct consequence of our result is the following corollary concerning the counting of \emph{minimal dominating sets} in a subclass of chordal graphs, called \textit{strongly chordal} graphs.
\begin{mcorollary}
	The number of minimal dominating sets of a strongly chordal graph is computable in polynomial time.
\end{mcorollary}

Besides the polynomial time algorithm, the main contribution of this article is the modification of the framework considered in \cite{Capelli17} in order to count minimal
models.  The techniques used in \cite{KanteLMNU13,GolovachHKKSV17} are based on structural restrictions and as shown in \cite{Capelli17} cannot work for $\beta$-acyclic
hypergraphs.  Instead, Capelli showed in \cite{Capelli17} how to construct, from the elimination ordering of a $\beta$-acyclic hypergraph associated to a boolean formula,
a circuit whose satisfying assignments correspond to the models of the boolean formula.  Such circuits are known as \emph{decision Decomposable Negation Normal Form} in
knowledge compilation.  While the technique allows to count the models of non-monotone formulas, it cannot be used to count the minimal models.  Indeed, the branchings of
the constructed circuit do not allow to control the minimality.  We overcome this difficulty by introducing the notion of \emph{blocked transversals}, which correspond
roughly to the minimal transversals of a sub-hypergraph that are transversals of the whole hypergraph.  We then show that blocked transversals can be used to control the
minimality in the construction of the circuit.  However, this control is only possible in the case of monotone Boolean formulas, corresponding to counting the minimal
transversals of $\beta$-acyclic hypergraphs.

Because of technical definitions, we postpone the details of the algorithm in Section \ref{sec:algo}.  The paper is organised as follows. Notations and some technical
definitions are given in Section \ref{sec:def}, while \emph{blocked transversals} and intermediate lemmas are given in Section \ref{sec:blocked-tr}.  The decomposition of
$\beta$-acyclic hypergraphs proposed in \cite{Capelli17} is refined in Section \ref{sec:decbeta} to take into account \emph{blocked transversals}.  Finally, the algorithm
is given in Section \ref{sec:algo} and we talk about its consequences for the counting of minimal dominating sets in Section \ref{sec:dom}.

\section{Definitions and notations}
\label{sec:def}

 The power set of a set $V$ is denoted by $2^V$, and its cardinal is denoted by $\#V$. For two sets $A$ and $B$, we let $A\setminus B$ denote the set $\{x\in A\mid x\notin B\}$. For a ground set $V$ and subsets $\calA_1,\ldots,\calA_k$ of $2^V$, with $k\geq 2$, we let
\begin{align*}
  \xcup_{1\leq i\leq k} \calA_i & := \begin{cases}  \emptyset  & \textrm{ if $\calA_i=\emptyset$ for some $1\leq i \leq k$},\\ 
  	\{T_1\cup \dots \cup T_k \mid \forall i\leq k, \ T_i\in \calA_i\}  
    &\textrm{ otherwise}. \end{cases} 
\end{align*}
If $k=2$, we write $\calA_1\xcup \calA_2$ as usual. 
For example, if $\calA_1=\{ \{a,b\}, \{ b,c\} \}$ and $\calA_2=\{\{ d\}, \{e,f \}\}$, then $\calA_1\xcup \calA_2 = \{ \{ a,b,d  \}, \{a,b,e,f \}, \{ b,c,d \}, \{ b,c,e,f\} \}$.
\subsection{Hypergraphs}

A hypergraph $\calH$ is a collection of subsets of a finite ground set. The elements of $\calH$ are called the {\em hyperedges} of $\calH$ and the {\em vertex set} of $\calH$
is $V(\calH) := \bigcup_{e \in \calH} e$. 

Given $S \subseteq V(\calH)$, we denote by $\calH[S]$ the {\em hypergraph induced by $S$}, that is, $\calH[S] := \{e \cap S \mid e \in \calH \}$.  Any subset $\calH'$ of
$\calH$ is called a \emph{sub-hypergraph} of $\calH$.
%
%For $x \in V(\calH)$, we write $\calH\setminus x$ for $\calH[V(\calH) \setminus \{x\}]$.
Observe that if there exists $e \in \calH$ such that $e \subseteq V(\calH)\setminus S$, then $\emptyset \in \calH[S]$. We do not enforce hypergraphs to have non-empty
edges or to be non-empty. However, a hypergraph with an empty edge may behave counter-intuitively. In the following definitions, we explicitly explain the extremal cases
where $\emptyset\in\cH$ or $\calH = \emptyset$.

Given a hypergraph $\cH$ and a subset $S\subseteq V(\cH)$ of its vertex set, we denote by $\cH(S)$ the set of hyperedges of $\cH$ containing at least one vertex in $S$,
that is, $\cH(S) := \{e \in \cH \mid S\cap e\neq \emptyset\}$; to ease notations, we write $\cH(x)$ instead of $\cH(\{x\})$ for $x\in V(\calH)$.  Observe that $\cH(x)$ is the set of hyperedges containing $x$. 
%We use this notation on sub-hypergraphs of $\cH$, but also on sets of (minimal) transversals of $\cH$.  
For the hypergraph $\cH$ shown in Figure \ref{fig:example}, $\cH(\{d,x\})$ is the hypergraph $\{\{b,x\},\{c,x\},\{c,d\}\}$.

Given a hypergraph $\calH$, a walk between two distinct edges $e_1$ and $e_k$ is a sequence $(e_1,x_1,e_2,x_2, \ldots, e_{k-1},x_{k-1},e_k)$ such that $x_i\in e_i\cap
e_{i+1}$ for all $1\leq i \leq k-1$. Notice that, $(e_k,x_{k-1},e_{k-1}, \ldots, x_2, e_2,x_1,e_1)$ is also a walk between $e_k$ and $e_1$. A maximal set of edges of
$\calH$ that are pairwise connected by a walk is called a \emph{connected component} of $\calH$. It is worth noticing that if $C_1,\ldots, C_k$ are the connected
components of $\calH$, then $V(C_i)\cap V(C_j)=\emptyset$ for distinct $i,j$ in $\{1,\ldots, k\}$. 

A hypergraph $\calH$ is said \emph{$\beta$-acyclic} if there exists an ordering $x_1,\ldots, x_n$ of $V(\calH)$ such that for each $1\leq i \leq n$, the set
$\{e\cap \{x_i,\ldots,x_n\}\mid e\in \calH,\ x_i\in e\}$ is linearly ordered by inclusion. Such an ordering is called a \emph{$\beta$-elimination ordering}. It is well-known that every
sub-hypergraph $\calH'$ of a $\beta$-acyclic hypergraph $\calH$ are $\beta$-acyclic as well (see for instance \cite{EiterG95}).

\subsection{Transversals}
\label{sec:trans}

Let $\calH$ be a hypergraph. A transversal for $\calH$ is a subset $T \subseteq V(\calH)$ such that for every $e \in \calH$, $T \cap e \neq \emptyset$.  We denote by
$\tr(\calH)$ the set of transversals of $\calH$. Observe that if $\emptyset \in \calH$, then $\tr(\calH) = \emptyset$ as for every $T \subseteq V(\calH)$,
$\emptyset \cap T = \emptyset$ so $T$ cannot be a transversal of $\calH$.  Finally, observe that if $\calH = \emptyset$, then $\tr(\calH) = \{\emptyset\}$.

A transversal $T$ of $\calH$ is {\em minimal} if and only if for every $x \in T$, it holds that $T \setminus \{x\} \notin \tr(\calH)$. A hyperedge $e$ such that
$e \cap T = \{x\}$ is said to be {\em private for $x$ \wrt $T$}. When $T$ is clear from the context, we may refer to such hyperedges as simply privates for $x$.  The following fact follows directly from the definitions:

\begin{fact}\label{fact:minT} $T$ is a minimal transversal of a hypergraph $\cH$ if and only if $T$ is a transversal of $\cH$ and each vertex $x\in T$ has a private.
\end{fact}
  
We denote by $\mtr(\calH)$ the set of minimal transversals of $\calH$. Again, observe
that if $\calH = \emptyset$ then $\mtr(\calH) = \{\emptyset\}$.

Figure \ref{fig:example} depicts a hypergraph together with its minimal transversals.  

\begin{figure}[!h]
	\centerline{\includegraphics[scale=0.95]{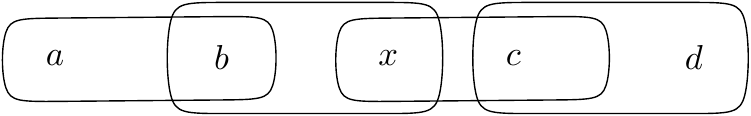}}
	\caption{A hypergraph $\cH=\{ \{a,b\}, \{b,x \}, \{x,c\} ,\{c,d\}\}$ having $4$ minimal transversals. We have $\mtr(\cH)=\{ \{ a,x,d\}, \{a,x,c \}, \{b,x,d\}, \{b,c\} \}$.}
\label{fig:example}
\end{figure}

It is worth noticing that since the sets $\tr(\cH)$ and $\mtr(\cH)$ are sets of subsets of $2^{V(\cH)}$, they may be seen as hypergraphs on $V(\cH)$. Thus, we will sometimes use the notations $\tr(\cH)(S)$
(resp. $\mtr(\cH)(S)$) which refer to the transversals (resp. minimal transversals) $T$ of $\cH$ such that $S \cap T \neq \emptyset$. 

\section{Blocked transversals}\label{sec:blocked-tr}

Our algorithm uses a dynamic programming approach by finding a relation between the number of minimal transversals of a $\beta$-acyclic hypergraph  $\cH$ with the number of minimal transversals of some specific subhypergraphs of $\cH$. However, it is not possible to directly relate these quantities together. To illustrate this fact, let $\cH$ be the hypergraph depicted in Figure \ref{fig:example}.  The minimal transversals of $\cH$ containing $x$ are
$\{ a,x,c\}, \{a,x,d \}$ and $\{b,x,d \}$.  Observe that removing $x$ from these sets directly yields  a minimal transversal of
$\cH\setminus \cH(x)=\{\{a,b\} ,\{c,d\} \}$.  However, adding $x$ to a minimal transversal of $\cH\setminus \cH(x)$ does not give necessarily a minimal transversal of $\cH$. 
For example, we have $\{b,c\} \in \mtr(\cH\setminus \cH(x))$, but $\{b,x,c\} \notin \mtr(\cH)$.

In fact, we can show in general that $T\cup x \in \mtr(\cH)$ if and only if $T$ is a minimal transversal of $\cH\setminus\cH(x)$ and $T$ is not a transversal of $\cH$.
Consequently, the number of minimal transversals of $\cH$ containing $x$ is  
\begin{align}\label{eq:1} \sharp \mtr(\calH\setminus \calH(x)) - \sharp (\tr(\cH)\cap \mtr(\calH\setminus \calH(x))). \end{align}
We can infer from this fact a recursive formula to compute $\sharp\mtr(\cH)$.
In the general case, using this formula as it is will lead to the computation of an exponential number of terms. In this section, we will make this relation more precise by introducing the notion of \emph{blocked transversal}. In the next section, we will show how we can use this relation to evaluate the number of minimal transversals of a $\beta$-acyclic hypergraphs  with only a polynomial number of intermediate values.

Given a hypergraph $\calH$, a sub-hypergraph $\calH' \subseteq \calH$ and $\calB \subseteq V(\calH)$, we define the {\em $\calB$-blocked transversals} of $\calH'$ to be
the transversals $T$ of $\calH'$ such that each vertex $x$ of $T$ has a private in $\calH' \setminus \calH'(\calB)$.  In particular, if $y$ is a vertex of $\calB$,
then $y$ cannot be in a $\calB$-blocked transversal of $\calH'$. Observe also that if $y \in V(\calH) \setminus V(\calH')$, then $y$ cannot be in a $B$-blocked
transversal of $\cH'$.

\begin{fact}\label{fact:blocked} Given a sub-hypergraph $\cH'$ of a hypergraph $\cH$ and $B\subseteq V(\cH)$, $T$ is a $B$-blocked transversal of $\cH'$ if and only if
  $T$ belongs to $\tr(\cH')\cap \mtr(\cH'\setminus \cH'(B))$.
\end{fact}

\begin{proof} If $T$ is a transversal of $\cH'$ such that each vertex $x$ of $T$ has a private in $\cH'\setminus \cH'(B)$, then $T\subseteq V(\cH'\setminus \cH'(B))$ and
  is a transversal of $\cH'\setminus \cH'(B)$. Thus by Fact \ref{fact:minT} $T$ is a minimal transversal of $\cH'\setminus \cH'(B)$.  Conversely, if $T$ is a minimal
  transversal of $\cH'\setminus \cH'(B)$, then by Fact \ref{fact:minT} each vertex $x$ has a private in $\cH'\setminus \cH'(B)$. If in addition it belongs to the set of
  transversals of $\cH'$, then it is a $B$-blocked transversal of $\cH'$ by definition.
\end{proof}

As an example, let $\cH$ be the hypergraph depicted in Figure \ref{fig:example}, $\cH'=\cH$ and $B=\{x\}$. The only $B$-blocked transversal of $\cH$ is
$\{b,c\}$. While $\{b,d\}$ is a minimal transversal of $\cH\setminus \cH(x)$, it is not a $B$-blocked transversal as the hyperedge $\{x,c\}$ does not intersect
$\{b,d\}$. 

\medskip

% In other words, a transversal $T$ of $\calH'$ is a $\calB$-blocked transversal if and only if for every $x \in T$, there exists $e \in \calH' \setminus \calH'(\calB)$
% such that $e \cap T = \{x\}$.
We call the set $\calH'(\calB)$ the {\em blocked hyperedges}. Intuitively, $\calH'(\calB)$ is the set of hyperedges that cannot be used as privates in a transversal of
$\cH'$. We denote by $\btr_{\calB}(\calH')$ the set of $\calB$-blocked transversals of $\calH'$. By Fact \ref{fact:blocked}, we have
\[ \btr_\calB(\calH') := \tr(\calH') \cap \mtr(\calH'\setminus \calH'(\calB)). \] Observe that by definition, $\mtr(\calH) = \btr_\emptyset(\calH)$. Moreover, if
$\calH(\calB) = \calH \neq \emptyset$, then $\btr_{\calH}(\calH) = \emptyset$ as $\mtr(\emptyset) = \{\emptyset\}$ and $\emptyset \notin \tr(\calH)$.  When $\calB=\{x\}$,
we denote $\btr_\calB(\calH)$ by $\btr_x(\calH)$. 

Given $S \subseteq V(\calH)$, we denote by $\tr(\calH,S) := \{T \in \tr(\calH) \mid T \subseteq S \}$. We extend this notation to $\mtr$ and $\btr$ as well.  The
following summarises observations about blocked transversals.

\begin{fact}\label{fact:prop-blocked} Let $\cH'$ be a sub-hypergraph of a hypergraph $\cH$ and $B,S\subseteq V(\cH)$. Then,
  \begin{enumerate}
  \item $\btr_\calB(\calH') = \btr_{\calB\cap V(\cH')}(\calH')$.
  \item $\btr_\calB(\calH',S)=\btr_\calB(\calH',S\setminus\calB)$.
  \item If $x\notin V(\calH')$, then $\btr_\calB(\calH',S)=\btr_\calB(\calH',S\setminus\{x\})$.
  \end{enumerate}
\end{fact}

Let us briefly explain how $B$-blocked transversals will be used in computing $\#\mtr(\cH)$. One checks easily that
$\#\mtr(\cH) = \#\mtr(\cH,V(\cH)\setminus \{x\}) + \#\mtr(\cH)(x)$. By Equation \ref{eq:1}, $\#\mtr(\cH)(x)=\#\mtr(\cH\setminus \cH(x)) - \#\btr_x(\cH)$.  Therefore,
$\#\mtr(\cH) = \#\btr_\emptyset(\cH) = \#\btr_\emptyset(\cH,V(\cH)\setminus \{x\}) + \#\btr_\emptyset(\cH\setminus \cH(x)) - \#\btr_x(\cH)$. We will show in the next
section that we can define, from the $\beta$-elimination ordering of a $\beta$-acyclic hypergraph $\cH$, sub-hypergraphs $\cH_1, \ldots, \cH_q\subseteq \cH$ and vertices
$x_1,\ldots,x_q$ such that, for all $1\leq i \leq q$, $\#\btr_{x_i}(\cH_i)$ and $\#\btr_\emptyset(\cH_i)$ can be computed in polynomial time if $\#\btr_{x_j}(\cH_j)$ and
$\#\btr_\emptyset(\cH_j)$ are known for all $j< i$. As a consequence, one can compute $\#\mtr(\cH)$ by classical dynamic programming for any $\beta$-acyclic
hypergraph.

The end of this section is dedicated to the proof of several crucial lemmas concerning recursive formulas for computing the number of blocked transversals, and that will
be useful in our algorithm. We start by describing the blocked transversals of a hypergraph having more than one connected component.

\begin{lemma}
  \label{lem:btrcc} Let $\calH$ be a hypergraph, $S,\calB \subseteq V( \calH)$ and $C_1,\ldots,C_k$ the connected components  of $\calH[S]$. For each $i\in[1,k]$, let $\calH_i=\{e\in\calH \mid e\cap S\in C_i \}$. We have:
\[ \btr_\calB(\calH,S) = \xcup_{i=1}^k \btr_{\calB}(\calH_i,S).\]
\end{lemma}
\begin{proof}
  Assume that $\emptyset \in \calH[S]$, it means that there exists a hyperedge $e\in\calH$ such that $e\cap S=\emptyset$. In this case, $\btr_\calB(\calH,S)=\emptyset$. Moreover, there exists
  $i\in [1,k]$ such that $C_i=\{\emptyset\}$ and $e\in\calH_i$. Thus, $\btr_{\calB}(\calH_i,S)=\emptyset$ and the equality holds in this case.

  Assume from now that $\emptyset\notin \calH[S]$.  Let $T \in \btr_\calB(\calH,S)$. We show that for all $i \leq k$, $T_i = T \cap V(\calH_i) \in \btr_{\calB}(\calH_i,S)$. Let $e \in \calH_i$. Since
  $e \in \calH$, we have $e\cap T \neq \emptyset$.  As $e\in\calH_i$, we have $e \subseteq V(\calH_i)$.  Thus $e\cap T_i\neq \emptyset$, that is, $T_i \in \tr(\calH_i,S)$.  Moreover, let $y \in T_i$.
  By definition of $T$, there exists $e \in \calH \setminus \calH(\calB)$ such that $e$ is private for $y$ \wrt $T$. Observe that we have $T_i\subseteq S \cap V(\calH_i) = V(C_i)$ since
  $T\subseteq S$. Thus, $y\in V(C_i)$ and $e\cap S \in C_i$, because $C_i$ is a connected component. We can conclude that $e\in \calH_i$ and $e$ is private to $y$ \wrt $T_i$ and
  $\calH_i \setminus \calH(\calB) = \calH_i \setminus \calH_i(\calB)$. Thus $T_i$ is a minimal transversal of $\calH_i \setminus \calH_i(\calB)$. That is $T_i \in \btr_{\calB}(\calH_i,S)$.

  Now let $T_1 \in \btr_{\calB}(\calH_1,S), \ldots, T_k \in \btr_{\calB}(\calH_k,S)$.  We show that $T = \bigcup_{i=1}^k T_i \in \btr_\calB(\calH,S)$.  Let $e \in \calH$. As
  $\calH = \bigcup_{i=1}^k \calH_i$, there exists $i$ such that $e \in \calH_i$.  Thus $e \cap T_i \neq \emptyset$ and thus $e \cap T \neq \emptyset$, that is, $T \in \tr(\calH)$. It remains to show
  that $T\in\mtr(\calH\setminus \calH(\calB))$. Let $y \in T$.  By definition of $T$, there exists $i$ such that $y \in T_i$.  Thus, there exists $e \in \calH_i \setminus \calH_i(\calB)$ that is
  private for $y$ \wrt $T_i$.  Moreover, since $\calH_i(\calB) = \calH_i \cap \calH(\calB)$, we know that $e \notin \calH(\calB)$.  As $C_1,\dots,C_k$ are the connected component of $\calH[S]$, we
  have that, for every $j \neq i$, $V(C_i)\cap V(C_j)=\emptyset$. Moreover, for all $\ell \leq k$, we have $S\cap V(\calH_\ell)=V(C_\ell)$. As $e\subseteq V(\calH_i)$ and
  $T_j\subseteq S \cap V(\calH_j)$, we have for every $j\neq i$:
  \[  e \cap T_j\subseteq V(\calH_i)\cap S\cap V(\calH_j) = V(C_i)\cap V(C_j)=\emptyset. \] 
  Thus, $e\cap T = e \cap T_i = \{y\}$. In other words, $e$ is private for $y$ \wrt $T$ and $\calH \setminus \calH(\calB)$. That is $T\in \btr_{\calB}(\calH)$.
\end{proof}

We recall that $\btr_B(\cH,S)(x)$ is the set of $B$-blocked transversals $T$ of $\cH$ such that $T\subseteq S$ and $x\in T$. The following lemma shows that for any
$B$-blocked transversal $T\subseteq S$ of $\cH$ containing $x$, we have $T\setminus x$ is a $B$-blocked transversal of $\cH\setminus \cH(x)$.

\begin{lemma}
\label{lem:btrinc}
  Let $\calH$ be a hypergraph, $S,\calB \subseteq V(\calH)$ and $x \in S$. 
  We have \[\btr_\calB(\calH,S)(x) \subseteq \{\{x\}\} \xcup \btr_{\calB}(\calH\setminus \cH(x),S\setminus\{x\}).\]
\end{lemma}
\begin{proof} Let $\calH_1 := \calH \setminus \calH(x)$.  Let $T \in \btr_\calB(\calH,S)(x)$. By definition, $x \in T$ and $T\subseteq S$, thus we only have to show that
  $T' = T \setminus \{x\} \in \btr_{\calB}(\calH_1)$. Let $e \in \calH \setminus \calH(x)$. Since $T$ is a transversal of $\calH$, there exists $y \in e \cap T$.
  Moreover, by definition, $x \notin e$, thus $y \in T'$, that is, $T'$ is a transversal of $\calH \setminus \calH(x)$. It remains to show that $T'$ is a minimal
  transversal of $\calH_1\setminus \calH_1(\calB) = \calH \setminus (\calH(\calB) \cup \calH(x))$. Let $y \in T'$. Since $T$ is a minimal transversal of
  $\calH \setminus \calH(\calB)$, there exists $e \in \calH \setminus \calH(\calB)$ such that $e$ is private for $y$ \wrt $T$. Since $x \in T$, we have $x \notin e$,
  otherwise $e$ would not be private for $y$ \wrt $T$. Thus $e \in \calH \setminus (\calH(\calB) \cup \calH(x))$, that is, $e$ is private to $y$ \wrt $T'$ in
  $\calH_1\setminus \calH_1(\calB)$. In other words, $T'$ is a minimal transversal of $\calH_1\setminus \calH_1(\calB)$ which concludes the proof.
\end{proof}

To complete the previous lemma, we show that for each $B$-blocked transversal $T\subseteq S$ of $\cH\setminus\cH(x)$, we have $T\cup \{x\}$ is a $B$-blocked transversal
of $\cH$ if and only if $T$ is not a $(B\cup \{x\})$-blocked transversal of $\cH\setminus (\cH(B)\cap \cH(x))$.

\begin{lemma}
\label{lem:btreq}
  Let $\calH$ be a hypergraph, $S,\calB \subseteq V(\calH)$ and $x \in S$. We have 
  \[ \big(\{\{x\}\} \xcup \btr_{\calB}(\calH_1,S\setminus\{x\})\big) \setminus \btr_{\calB}(\calH,S)(x) = \{\{x\}\} \xcup \btr_{\calB \cup \{x\}}(\calH_2,S\setminus\{x\})  \] where 
 $\calH_1 := \calH \setminus \calH(x)$ and $\calH_2 := \calH \setminus (\calH(\calB) \cap \calH(x))$. 
\end{lemma}

\begin{proof} We prove the lemma by proving first the left-to-right inclusion (Claim~\ref{clm:lir}) and then the right-to-left inclusion (Claim~\ref{clm:ril}). But first,
  notice that $\calH_1\setminus\calH_1(\calB)=\calH_2\setminus\calH_2(\calB \cup \{x\})=\calH\setminus(\calH(\calB)\cup \calH(x))$ since
  $\calH(\calB)\cup\calH(x)=\calH(\calB\cup\{x\})$. 

\begin{claim}
\label{clm:lir}
For every $T \in \big(\{\{x\}\} \xcup \btr_{\calB}(\calH_1,S\setminus\{x\})\big) \setminus \btr_{\calB}(\calH,S)(x)$, we have $ T \setminus \{x\} \in \btr_{\calB \cup \{x\}}(\calH_2,S\setminus\{x\})$. 
\end{claim}
\begin{proof}
  We start by proving that $T' =T\setminus \{x\}$ is in $\tr(\calH_2)$. Assume towards a contradiction that $T'\notin \tr(\calH_2)$, \ie, there exists $e \in \calH_2$
  such that $e\cap T'=\emptyset$. We prove that it implies $T \in \btr_B(\calH,S)(x)$.  First, observe that $T\in \tr(\calH)$, since
  $T'\in \tr(\calH_1)=\tr(\calH\setminus\calH(x))$ and $T=T'\cup \{x\}$.  Thus, we have $e\cap T=\{x\}$ and $e\in\cH(x)$. As
  $e\in\calH_2=\calH\setminus (\calH(\calB)\cap \calH(x))$, we have $e\in \calH\setminus\calH(\calB)$ and then $e$ is a private hyperedge for $x$ \wrt $T$ and
  $\calH\setminus\calH(\calB)$.  Furthermore, every vertex in $T'$ has a private hyperedge \wrt $T'$ and
  $\calH_1\setminus\calH_1(\calB)=\calH\setminus(\calH(\calB)\cup \calH(x))$ since $T'\in \mtr(\calH_1\setminus\calH_1(\calB))$.  Thus, every vertex in $T'$ has a private
  hyperedge \wrt $T$ and $\calH\setminus\calH(\calB)$.  As $T\in\tr(\calH)$, we can conclude that $T\in\mtr(\calH\setminus\calH(\calB))$.  Finally, we have
  $T \subseteq S$ by assumption. Therefore $T \in \btr_\calB(\calH,S)(x)$ which is a contradiction. Thus, $T'\in \tr(\calH_2)$.

%By definition of $\calH_2=\calH \setminus (\calH(\calB) \cap \calH(x))$. We first assume that $x \notin e$. Therefore $e \notin \calH(x)$ and then $e \in \calH_1$. Since $T \in \btr_{\calB_1}(\calH_1)$, there exists $y \in T$ such that $y \in e$, that is, $e$ is covered by $T$. As $x \notin e$, we have $y \neq x$ and then $e \cap T' \neq \emptyset$, that is, $e$ is covered by $T'$.
%Now assume that $x \in e$. Assume toward a contradiction that $e\cap T' = \emptyset$. In this case, $e \cap T = \{x\}$, that is, $e$ is private for $x$ \wrt $T$. Now let $y \in T$ be such that $y \neq x$. By definition of $T$, there exists $g \in \calH_1 \setminus \calB_1$ such that $g \cap T = \{y\}$. In particular, $g \in \calH$. Moreover, $g \notin \calB$ since $g \notin \calH(x) \cup \calB_1$. Thus $g \in \calH \setminus \calB$ meaning that $g$ is private for $y$ \wrt $T$. In other words, every vertex in $T$ has a private hyperedge in $\calH\setminus\calB$. Furthermore, $T\in\tr(\calH)$ since $T'\in\tr(\calH\setminus\calH(x))$ and $x\in T$. Thus, $T\in \btr_{\calB}(\cH,x)$ which is a contradiction. Thus $e \cap T \neq \emptyset$, which finishes the proof that $T' \in \tr(\calH_2)$.

We now prove that $T' \in \mtr(\calH_2 \setminus \calH_2(\calB \cup \{x\}))$, that is, we prove the minimality of $T'$ in $\calH_2 \setminus \calH_2(\calB \cup \{x\})$. Let $y \in T'$. Since $T \in \btr_{\calB}(\calH_1)$, there exists $f \in \calH_1 \setminus \calH_1(\calB)$ such that $f \cap T = \{y\}$. Since $\calH_1 \setminus \calH_1(\calB) = \calH_2\setminus \calH_2(\calB \cup \{x\})$, every $y \in T'$ have a private hyperedge in $\calH_2 \setminus \calH_2(\calB \cup \{x\})$, that is $T' \in \mtr(\calH_2 \setminus \calH_2(\calB \cup \{x\}))$. As $T'\subseteq S\setminus\{x\}$, we can conclude that $T'\in \btr_{\calB \cup \{x\}}(\calH_2,S\setminus\{x\})$.
\end{proof}
\begin{claim}
\label{clm:ril}
For every $T \in \{x\} \bigotimes \btr_{\calB \cup \{x\}}(\calH_2,S\setminus\{x\})$, we have 
$T \in \{x\} \bigotimes \btr_{\calB}(\calH_1,S\setminus\{x\}) \setminus \btr_{\calB}(\calH,S)(x)$.
\end{claim}

\begin{proof} We start by proving that $T' = T \setminus \{x\}$ is in $\btr_{\calB}(\calH_1,S\setminus\{x\})$. First, we show that $T'$ is a transversal of $\calH_1$. Let
  $e \in \calH_1$. By definition of $\calH_1$, $x \notin e$, thus $e \in \calH_2$ as well. Therefore $e \cap T' \neq \emptyset$.  We now prove that $T'$ is minimal in
  $\calH_1 \setminus \calH_1(\calB)$. As $T'\in\mtr(\calH_2\setminus\calH_2(\calB \cup \{x\}))$, every vertex in $T'$ has a private hyperedge in
  $\calH_2\setminus\calH_2(\calB \cup \{x\})$. Moreover, recall that $\calH_2\setminus\calH_2(\calB \cup \{x\})=\calH_1\setminus\calH_1(\calB)$. Thus, $T'$ is minimal in
  $\calH_1 \setminus \calH_1(\calB)$. As $T'\subseteq S\setminus\{x\}$, we have $T'\in \btr_{\calB}(\calH_1,S\setminus \{x\})$.

We finish the proof by showing that $T \notin \btr_\calB(\calH,S)(x)$. In order to prove it, we show that there is no private hyperedge for $x$ \wrt $T$ and $\calH \setminus
\calH(\calB)$. Indeed, since $T'\in\tr(\calH_2)$, every hyperedge in $\calH_2$ contains a vertex in $T'$. By definition of $\calH_2$, we have $\calH \setminus
\calH(\calB)\subseteq\calH_2$, thus for every hyperedge $e$ in $\calH \setminus \calH(\calB)$, we have $e\cap T \neq \{x\}$, \ie, $e$ is not a private for $x$.
\end{proof}

By Claim \ref{clm:lir}  and Claim \ref{clm:ril} we can conclude the lemma. 
\end{proof} 

Finally, we characterise the number of $B$-blocked transversals of a hypergraph that do not contain a given vertex. We use the symbol $\uplus$ for the disjoint union of
sets. 
\begin{lemma}
  \label{lem:btrnotx} Let $\calH$ be a hypergraph, $S,\calB \subseteq V(\calH)$ and $x \in S$. We have 
  \[ \btr_{\calB}(\calH,S) = \btr_{\calB}(\calH,S)(x) \uplus  \btr_{\calB}(\calH, S\setminus\{x\}). \]
\end{lemma}

\begin{proof} Let $T\subseteq S$ be a $B$-blocked transversal of $\cH$. Thus, either $x\in T$ and then $T\in \btr_B(\cH,S)(x)$, or $x\notin T$ and then $T\subseteq
  S\setminus \{x\}$, \ie, $T\in \btr_B(\cH,S\setminus \{x\})$. Since, the two cases are exclusive, we can conclude that $\btr_B(\cH,S)$ is the disjoint union of
  $\btr_B(\cH,S)(x)$ and of $\btr_B(\cH,S\setminus \{x\})$.
\end{proof}

% Observe that Lemma~\ref{lem:btrinc} and Lemma~\ref{lem:btreq} characterise the $B$-blocked transversals of $\calH$ containing $x$, while Lemma~\ref{lem:btrnotx}
% characterises those that do not contain $x$ anymore. 
A direct consequence of Lemma~\ref{lem:btrinc}, Lemma~\ref{lem:btreq} and Lemma~\ref{lem:btrnotx} is the
following equality characterising the number of $B$-blocked transversals of $\calH$.  This will be a crucial step in our dynamic programming scheme.

\begin{theorem}
\label{thm:inductivestep}
Let $\calH$ be a hypergraph, $S,\calB \subseteq V(\calH)$ and $x \in S$. We have  
\[ \#\btr_{\calB}(\calH,S) = \#\btr_{\calB}(\calH,S\setminus\{x\}) + \#\btr_{\calB}(\calH_1,S\setminus\{x\}) - \#\btr_{\calB \cup \{x\}}(\calH_2,S\setminus\{x\}) \] 
where $\calH_1 := \calH \setminus \calH(x)$ and $\calH_2 := \calH \setminus (\calH(\calB) \cap \calH(x))$.
\end{theorem}

\begin{proof} By Lemma \ref{lem:btrnotx}, $\#\btr_{\calB}(\calH,S) = \#\btr_{\calB}(\calH,S\setminus\{x\}) + \#\btr_{\calB}(\calH,S)(x)$. By Lemma \ref{lem:btrinc} and
  Lemma \ref{lem:btreq}, $\{\{x\}\}\otimes \btr_B(\cH_1,S\setminus \{x\}) = \btr_B(\cH,S)(x)\uplus \{\{x\}\}\otimes \btr_{B\cup \{x\}}(\cH_2,S\setminus \{x\})$. Hence,
  $\#\btr_B(\cH,S)(x) = \#\btr_B(\cH_1,S\setminus \{x\}) - \#\btr_{B\cup \{x\}}(\cH_2,S\setminus \{x\})$. Therefore, the claimed equality holds.
\end{proof}

\section{Counting the minimal transversals of $\beta$-acyclic hypergraphs}
\label{sec:count}

In this section, we fix $\calH$ a $\beta$-acyclic hypergraph, $\leq$ a $\beta$-elimination ordering of its vertices and we let $\leq_\calH$ the induced lexicographic
ordering on the hyperedges, \ie, $e\leq_{\calH} f$ if $\min((e\setminus f)\cup (f\setminus e))\in e$.  We denote by $\calH_e^x$ the sub-hypergraph of $\cH$ formed by the
hyperedges $f \in \calH$ such that there exists a walk from $f$ to $e$ going only through hyperedges smaller than $e$ and vertices smaller than $x$. For an example, take
the $\beta$-elimination ordering $a,b,x,c,d$ of the hypergraph in Figure \ref{fig:example} and the induced ordering $\{a,b\}, \{b,x\},\{x,c\},\{c,d\}$ on $\cH$. For
$e=\{x,c\}$, the hypergraph $\cH_e^x$ is composed of the hyperedges $\{b,x\}$ and $\{x,c\}$.

For a vertex $x$ of $V(\calH)$, we write $[\leq x]$, $[< x]$ and $[\geq x]$ for, respectively, $\{y \in V(\cH) \mid y \leq x\}$, $\{y \in V(\cH) \mid y \leq x \wedge
y\neq x\}$ and $\{y \in V(\cH) \mid x \leq y\}$. Moreover, we write $\cH[\leq x]$, $\cH[< x]$ and $\cH[\geq x]$ instead of, respectively, $\cH\big[[\leq x]\big]$,
$\cH\big[[< x]\big]$ and $\cH\big[[\geq x]\big]$.

%we write $\calH[\leq x]$, $\calH[\geq x]$, $\calH[<x]$ and $\calH[>x]$ for, respectively, $\calH[\{y\in V(\calH)\mid y \leq x\}]$,$\calH[\{y\in V(\calH)\mid x \leq y\}]$, $\calH[\{y\in V(\calH)\mid y < x\}]$ and $\calH[\{y\in V(\calH)\mid x < y\}]$.

\subsection{Decomposition of $\beta$-acyclic hypergraphs}
\label{sec:decbeta}

 The following two lemmas have been proven in~\cite[Section III-A]{Capelli17}.

\begin{lemma}[Theorem 3 in \cite{Capelli17}]
  \label{lem:hexvar} For every hyperedge $e \in \calH$, and $x \in V(\calH)$, we have $V(\calH_e^x) \cap [\geq x] \subseteq e$.
\end{lemma}
\begin{lemma}[Lemma 2 in \cite{Capelli17}]
  \label{lem:hexinc} Let $e$ and $f$ be two hyperedges of $\calH$ such that $e \leq_\calH f$, and let $x$ and $y$ be vertices of $\cH$ such that $x \leq y$. If
  $V(\calH_e^x) \cap V(\calH_f^y) \cap [\leq x] \neq \emptyset$, then $\calH_e^x \subseteq \calH_f^y$.

\end{lemma}

We prove a lemma on the decomposition of $\calH_e^x$ graphs that will be used with Lemma \ref{lem:btrcc} to propagate the dynamic programming algorithm.

\begin{lemma}
	\label{lem:betadecompo}
	Let $x\in V(\calH)$, $e \in \calH$ and $S \subseteq [\geq x]$. Let 
        \begin{align*} 
\calH' &:= \begin{cases} \calH_e^x & \textrm{if $S = \emptyset$},\\  \calH_e^x \setminus \left(\bigcap_{w\in
          S}\calH_e^x(w)\right) & \textrm{otherwise}. \end{cases}
        \end{align*}
        For every connected component $C$ of $\calH'[<x]$ different from $\{\emptyset\}$, there exists $y < x$ and $f \leq_\calH e$ such that $C = \calH_{f}^y[\leq y]$
        and $\calH_f^y=\{g\in \calH' \mid g\cap [<x] \in C \}$.
\end{lemma}

\begin{proof}	
	 Let $y = \max(V(C))$ and $f = \max\{g\in \calH' \mid g\cap [<x] \in C \}$.  We show that $\calH_f^y=\{g\in \calH' \mid g\cap [<x] \in C \}$. 
	 
         First, we observe that $\calH_f^y\subseteq \calH'$. If $S=\emptyset$, it follows from Lemma~\ref{lem:hexinc} because, in this case, $\calH'=\calH_e^x$. Suppose
         that $S\neq \emptyset$, by definition of $\calH'$, we have $S\nsubseteq f$. Moreover, by Lemma~\ref{lem:hexvar}, we have that
         $V(\calH_f^y)\cap [\geq y] \subseteq f$. As $S\subseteq [\geq x]$ and $x>y$, we have $S\subseteq [\geq y]$. Thus $S\nsubseteq V(\calH_f^y)$ and for all
         $g\in \calH_f^y$, we have $S\nsubseteq g$ since $g\subseteq V(\calH_f^y)$. We can conclude that $\calH_f^y\subseteq \calH'$.

	Now, we prove that every $g\in \calH_f^y$, we have $g \cap [<x] \in C$.
 Let $g\in \calH_f^y$. By definition of $\calH_f^y$ and because $\calH_f^y\subseteq \calH'$, there exists a path $P$ from $f$ to $g$ going only through vertices smaller than $y$ and hyperedges smaller than $f$ in $\calH'$. As $y<x$, we can conclude that $f\cap[<x]$ is connected to $g\cap [<x]$ in $\calH'$, \ie, $g \cap [<x] \in C$. 
	In other words, we have $\calH_f^y \subseteq \{g\in \calH' \mid g\cap [<x] \in C \}$.
	
	It remains to prove the other inclusion. Let $g \in \calH'$ with $g\cap [<x] \in C$. Since $C$ is a connected component of $\calH'[<x]$, there exists a path $P$ from $f\cap [<x]$ to $g\cap[<x]$. By the maximality of $y$ and $f$, $P$ goes only through vertices smaller than $y$ and hyperedges smaller than $f$. We can construct from $P$ a path $P'$ from $f$ to $g$ in $\calH'$ going through vertices smaller than $y$ and hyperedges smaller than $f$. As $\calH'\subseteq \calH$, we can conclude that $g\in\calH_f^y$ and thus, $\calH_f^y=\{g\in \calH' \mid g\cap [<x] \in C \}$. 
	Finally, observe that $C =\calH_f^y[<x]= \calH_{f}^y[\leq y]$ since $y=\max(V(C))$.
\end{proof}

\subsection{The algorithm}
\label{sec:algo}

In this subsection, we describe the dynamic programming algorithm we use to count the number of minimal transversals of a $\beta$-acyclic hypergraph.  We denote by $x_1$ the smallest element of $\leq$. 

Our goal is to compute $\#\btr_\emptyset(\calH_e^x,[\leq x])$ and $\#\btr_w(\calH_e^x,[\leq x])$ for every $e \in \calH$, $x \in V(\calH)$ and $w \in V(\calH)$ such that
$x < w$. Observe that it is enough for computing the number of minimal transversals of $\calH$ as $\#\mtr(\calH) = \#\btr_\emptyset(\calH_{e_m}^{x_n},[\leq x_n])$
where $e_m$ is the maximal hyperedge for $\leq_\calH$ and $x_n$ is the maximal vertex for $\leq$. Indeed, we have $\calH_{e_m}^{x_n} = \calH$ and $[\leq x_n]=V(\calH)$, thus
$\btr_\emptyset(\calH_{e_m}^{x_n},[\leq x_n]) = \btr_\emptyset(\calH) = \mtr(\calH)$.

The propagation of the dynamic programming works as follows: we use Theorem~\ref{thm:inductivestep} to reduce the computation of $\#\btr_\calB(\calH_e^x,[\leq x])$ to the
computation of $\#\btr$ for several hypergraphs that do not contain $x$. We then use Lemma~\ref{lem:betadecompo} to show that these hypergraphs can be decomposed into
disjoint hypergraphs of the form $\calH_f^y$ for $f \leq_\calH e$ and $y < x$ which allows us to compute $\#\btr_\calB(\calH_e^x,[\leq x])$ from precomputed values of the
form $\#\btr_{\calB'}(\calH_f^y,[\leq y])$, where $\calB'\in \{\calB, \{x\}\}$.

Before continuing, let us give a high-level description of the algorithm. For each $1\leq i \leq n$ and each $1\leq j\leq m$, let $tab[i,j,0]$ be
$\#\btr_\emptyset(\cH_{e_j}^{x_i},[\leq x_i])$, and for each $\ell>i$, let $tab[i,j,\ell]$ be $\#\btr_{x_\ell}(\cH_{e_j}^{x_i},[\leq x_i])$. Because the number of minimal
transversals of $\cH$ is $\#\btr_\emptyset(\cH_{e_m}^{x_n},[\leq x_n])$, it is enough to show how to compute $tab[n,m,0]$ in polynomial time. The following is a high
level description of the algorithm which computes $tab[i,j,\ell]$, for all $1\leq i\leq n$, $1\leq j\leq m$ and $\ell\in \{0,i+1,\ldots,n\}$.
  \begin{framed}
\vspace{-5mm}
\begin{small}
\begin{tabbing}
{\bf Algorithm} {\sf CountMinTransversals}$(\cH)$\\
\ \ \ \ \ \ \ $\cH$: a $\beta$-acyclic hypergraph\\
1. Let $x_1,\ldots,x_n$ a $\beta$-elimination ordering of $\cH$.\\
2. Let $e_1,\ldots,x_m$ the induced lexicographic ordering on $\cH$.\\
3. Precompute $\cH_{e_j}^{x_i},[\leq x_i]$ for every $i \leq n,j \leq m$. \\
4. {\bf for} $1\leq i\leq n$ and $i<\ell \leq n$ {\bf do}\\
5. \ \ \ {\bf for} $1\leq j \leq m$ {\bf do}\\
6. \ \ \ \ \ \ Compute $tab[i,j,0]$ from the recursive formula of $\#\btr_\emptyset(\cH_{e_j}^{x_i},[\leq x_i])$.\\
7. \ \ \ \ \ \ Compute $tab[i,j,\ell]$ from the recursive formula of $\#\btr_{x_\ell}(\cH_{e_j}^{x_i},[\leq x_i])$.\\
8. \ \ \ {\bf end for}\\
9. {\bf end for}\\
10. {\bf return} $tab[n,m,0]$
\end{tabbing}
\end{small}
\vspace{-5mm}
\end{framed}

In order to ease the presentation, the computation of $\#\btr_\emptyset(\cH_e^x,[\leq x])$ and that of $\#\btr_w(\cH_e^x,[\leq x])$ are separated, even though many of the arguments are similar.

\subsubsection{Base cases.} \label{subsubsec:base} We observe that for every $e \in \calH$, $\calH_e^{x_1}[\leq x_1]$ is either equal to $\{x_1\}$ or $\{\emptyset\}$. Thus, for every $e \in \calH$ and $w \in V(\calH)$ such that $w > x_1$, we can compute $\#\btr_\emptyset(\calH_e^{x},[\leq x_1])$ and $\#\btr_{w}(\calH_e^{x},[\leq x_1])$ in time $O(1)$.

\subsubsection{Computing $\#\btr_\emptyset(\calH_e^{x},[\leq x])$ by dynamic programming.} \label{subsubsec:empty}

We start by explaining how we can compute $\#\btr_\emptyset(\calH_e^{x},[\leq
x])$ in polynomial time if the values $\#\btr_\emptyset(\calH_f^{y},[\leq y])$ and $\#\btr_w(\calH_f^{y},[\leq y])$ have been precomputed for $f \leq_\calH e$ and $y < x$, $y<w$. 

We start by applying Theorem~\ref{thm:inductivestep}.
\[ 
\begin{aligned}
\#\btr_\emptyset(\calH_e^{x},[\leq x]) & =  \#\btr_\emptyset(\calH_e^{x}, [< x])  \\
& + \#\btr_{\emptyset}(\calH_e^{x} \setminus \calH_e^{x}(x) , [< x] )  \\
& - \#\btr_{x}(\calH_e^{x}, [< x]).
\end{aligned}
\]

Now, let $C_1, \ldots, C_k$ be the connected components of $\calH_e^{x}[< x]$. If there exists $i$ such that $C_i=\{\emptyset\}$, then
$\#\btr_\emptyset(\calH_e^{x}, [< x]) = \#\btr_{x}(\calH_e^{x}, [< x]) = 0$.  Otherwise, by applying Lemma~\ref{lem:betadecompo} with $S=\emptyset$, there exists, for
each $1\leq i \leq k$, $y_i < x$ and $f_i \leq_\calH e$ such that $\calH_{f_i}^{y_i}=\{g\in \calH_e^f \mid g\cap [<x] \in C_i\}$ and $C_i = \calH_{f_i}^{y_i}[\leq
y_i]$. By Lemma~\ref{lem:btrcc},
\begin{align*}
\#\btr_\emptyset(\calH_e^{x}, [< x]) & = \prod_{i=1}^k \#\btr_{\emptyset}(\calH_{f_i}^{y_i} , [\leq y_i]),\\
\#\btr_{x}(\calH_e^{x}, [< x])  = &
\prod_{i=1}^k \#\btr_{x}(\calH_{f_i}^{y_i} , [\leq y_i]) .
\end{align*}

%Moreover, for all $i$, $\#\btr_{\{x\} \cap V(\calH_{f_i}^{y_i})}(\calH_{f_i}^{y_i} , [\leq y_i])$ is precomputed since $ \{x\} \cap V(\calH_{f_i}^{y_i})$ is either equal to $\{x\}$ or $\emptyset$. 

We now show how to decompose $\#\btr_{\emptyset}(\calH_e^{x} \setminus \calH_e^{x}(x) , [< x] )$ into a product of precomputed values. Let $D_1,\ldots,D_l$ be the
connected components of $\calH_e^{x} \setminus \calH_e^{x}(x)[<x]$. If there exists $i$ such that $D_i=\{\emptyset\}$, then
$\#\btr_{x}(\calH_e^{x} \setminus \calH_e^{x}(x), [< x]) =0$.  Otherwise, if we apply Lemma~\ref{lem:betadecompo} with $S = \{x\}$, then there exists, for each
$1\leq i \leq l$, $y_i < x$ and $f_i \leq_\calH e$ such that $\calH_{y_i}^{f_i}= \{g\in \calH_e^{x} \setminus \calH_e^{x}(x) \mid g\cap [<x] \in D_i \}$ and
$D_i = \calH_{y_i}^{f_i}[\leq y_i]$. We can thus conclude by Lemma~\ref{lem:btrcc} that
\[ \#\btr_{\emptyset}(\calH_e^{x} \setminus \calH_e^{x}(x) , [< x] ) = \prod_{i=1}^l \#\btr_{\emptyset}(\calH_{f_i}^{y_i} , [\leq y_i]).\]

Therefore, if $\#\btr_{\emptyset}(\calH_{f}^{y} , [\leq y])$ and $\#\btr_{x}(\calH_{f}^{y} , [\leq y])$ have already been computed for every $f <_\calH e$ and $y \leq x$,
we can compute $\#\btr_\emptyset(\calH_e^{x},[\leq x])$ with at most $3 \times |\calH_e^x|$ additional multiplications and $3$ additions.

\subsubsection{Computing $\#\btr_w(\calH_e^{x},[\leq x])$ by dynamic programming.}\label{subsubsec:block} Let $x \leq w$. By Theorem~\ref{thm:inductivestep}, we have:
\[ 
\begin{aligned}
\#\btr_w(\calH_e^{x},[\leq x]) & = \#\btr_{w}(\calH_e^{x},[< x])  \\ 
& + \#\btr_{w}(\calH_e^{x}\setminus \calH_e^x(x),[< x])  \\
& - \#\btr_{\{x,w\}}(\calH_e^{x}\setminus (\calH_e^x(x) \cap \calH_e^x(w)), [< x] ).
\end{aligned}
\]

We start by explaining how to compute $\#\btr_{w}(\calH_e^{x},[< x]) $. Let $C_1,\dots,C_k$ be the connected components of $\calH_e^x[<x]$. If there exists $i$ such that
$C_i=\{\emptyset\}$, then $\#\btr_{w}(\calH_e^{x},[< x])=0$.  Otherwise, by Lemma~\ref{lem:betadecompo} with $S=\{w\}$, there exists, for each $1\leq i\leq k$,
$f_i \leq_\calH e$ and $y_i < x$ such that $\calH_{f_i}^{y_i}= \{g\in \calH_e^{x} \mid g\cap [<x]\in C_i \}$ and $C_i = \calH_{f_i}^{y_i}[\leq y_i]$. By
Lemma~\ref{lem:btrcc}, we can conclude that 
\[ \#\btr_{w}(\calH_e^{x},[< x]) = 
\prod_{i=1}^k \#\btr_{w}(\calH_{f_i}^{y_i}, [\leq y_i]). \]

%Moreover, for all $i$, $\#\btr_{\{w\}\cap V(\calH_{f_i}^{y_i})}(\calH_{f_i}^{y_i}, [\leq y_i])$ is a precomputed term since $\{w\}\cap V(\calH_{f_i}^{y_i})$ is either equal to $\{w\}$ or $\emptyset$.

We now explain how to compute $ \#\btr_{w}(\calH_e^{x}\setminus \calH_e^x(x),[< x])$. Let $D_1,\dots,D_l$ be the connected components of $(\calH_e^{x}\setminus \calH_e^x(x))[<x]$. 
If there exists $i$ such that $D_i=\{\emptyset\}$, then $\#\btr_{w}(\calH_e^{x}\setminus \calH_e^x(x),[< x]) = 0$. 
Otherwise, by applying Lemma~\ref{lem:betadecompo} with $S = \{x\}$, it follows that for every $1\leq i\leq l$, there exists $y_i < x$ and $f_i \leq_\calH e$ such that $\calH_{f_i}^{y_i} = \{g\in \calH_e^{x}\setminus \calH_e^x(x) \mid g\cap [<x]\in D_i \}$ and $D_i = \calH_{f_i}^{y_i}[\leq y_i]$.
Thus, from  Lemma~\ref{lem:btrcc},
\[  \#\btr_{w}(\calH_e^{x}\setminus \calH_e^x(x),[< x])  = 
\prod_{i=1}^l \#\btr_{w}(\calH_{f_i}^{y_i}, [\leq y_i]).\]
%As $\{w\}\cap V(\calH_{f_i}^{y_i})$ is either equal to $\{w\}$ or $\emptyset$, which has been precomputed.

Finally, we explain how to decompose $\#\btr_{\{x,w\}}(\calH_e^{x}\setminus (\calH_e^x(x) \cap \calH_e^x(w)), [< x] )$ into a product of pre-computed values. To ease the
notation, we denote $\calH_e^{x}\setminus (\calH_e^x(x) \cap \calH_e^x(w))$ by $\calH'$.  Let $K_1,\dots,K_r$ be the connected components of $\calH'[<x]$. If there exists
$i$ such that $K_i=\{\emptyset \}$, then $\#\btr_{\{x,w\}}(\calH', [< x] )=0$.  Otherwise, by Lemma~\ref{lem:betadecompo} applied with $S = \{x,w\}$, we have that for every
$i$, there exists $y_i < x$ and $f_i \leq_\calH e$ such that $\calH_{f_i}^{y_i}= \{g\in \calH' \mid g\cap [<x] \in K_i \}$ and $K_i = \calH_{f_i}^{y_i}[\leq y_i]$. By
Lemma~\ref{lem:btrcc}, we have:
\[ \#\btr_{\{x,w\}}(\calH_e^{x}\setminus (\calH_e^x(x) \cap \calH_e^x(w)), [< x] ) =
 \prod_{i=1}^r \#\btr_{\{x,w\}\cap V(\calH_{f_i}^{y_i})}(\calH_{f_i}^{y_i},[\leq y_i]). \]

\begin{claim}
\label{clm:bcapci3}
  For every $i \leq p$, $ \{x,w\}\cap V(\calH_{f_i}^{y_i})\neq \{x,w\}$.
\end{claim}
\begin{proof}
Assume towards a contradiction that $\{x,w\}\cap V(\calH_{f_i}^{y_i})= \{x,w\}$. Recall that $\calH_{f_i}^{y_i}[\leq y_i]$. By Lemma \ref{lem:hexvar}, $\{x,w\}\subseteq V(\calH_{f_i}^{y_i})$ implies $\{x,w\}\subseteq f$. Thus, we have $f\in \calH_e^x(x)\cap\calH_e^x(w)$. This is a contradiction, since $f\in \calH_{f_i}^{y_i} \subseteq \calH_e^x \setminus (\calH_e^x(w) \cap \calH_e^x(x)) $.
\end{proof}

Thus, $\{x,w\}\cap V(\calH_{f_i}^{y_i})$ equals either $\{x\}$, or $\{w\}$ or $\emptyset$ by Claim \ref{clm:bcapci3}. That is, we can compute
$\#\btr_{\{x,w\}}(\calH_e^{x}\setminus (\calH_e^x(x) \cap \calH_e^x(w)), [< x] )$ from precomputed terms.

We can conclude that, if $\#\btr_{\emptyset}(\calH_{f}^{y} , [\leq y])$ and $\#\btr_w(\calH_{f}^{y} , [\leq y])$ have already been computed for every $f <_\calH e$ and
$y \leq w$, we can compute $\#\btr_w(\calH_e^{x},[\leq x])$ with at most $3 \times |\calH_e^x|$ additional multiplications and $3$ additions.

\medskip

It is easy to see that a straightforward greedy algorithm can be used to compute a $\beta$-elimination ordering in polynomial time (see \cite{PaigeT87} for a better
algorithm due to Paige and Tarjan). Moreover, the dynamic programming algorithm describes above computes at most $O(n^2|\calH|)$ terms and each of them can be computed
from the others with a polynomial number of arithmetic operations. Finally, all these terms can be bounded by $2^n$ since they are all the cardinals of some collection of
subsets of the vertices. Thus these arithmetic operations can be done in polynomial time in the size of the input. It follows.

\begin{theorem}[Main Theorem]\label{thm:main}
  Let $\calH$ be a $\beta$-acyclic hypergraph. One can compute in polynomial time the number of minimal transversals of $\cH$.
\end{theorem}

\section{Applications to the counting of Dominating Sets}
\label{sec:dom}

We refer to \cite{Diestel05} for our graph terminology. For a graph $G$, let $V(G)$ be its set of vertices and $E(G)$ be its set of edges. For a vertex $x$ of a graph
$G$, let $N(x)$ be the set of neighbours of $x$ and we let $N[x]$ be the set $\{x\}\cup N(x)$. The \emph{closed neighbourhood hypergraph} of $G$, denoted by $\cN[G]$, is
the hypergraph $\{N[x]\mid x\in V(G)\}$. A \emph{dominating set} in a graph $G$ is a transversal of $\cN[G]$. \textsc{Dominating Set} problems are classic and
well-studied graph problems, and has applications in many areas such as networks and graph theory \cite{HaynesHS98}. 

In \cite{KanteLMN14} the authors reduce the \textsc{Hypergraph Dualisation} problem into the enumeration of minimal dominating sets, showing that the two problems are equivalent in the area of enumeration problems
(a fact already established in the case of optimisation).  The reduction indeed shows that the counting versions are equivalent (under Turing reductions), and such a
reduction is of big interest because it allows to study counting and enumeration problems associated with the \textsc{Hypergraph Dualisation} in the perspectives of graph
theory, where tools had been developed to tackle combinatorial problems.

Despite the broad application of counting the minimal dominating sets in (hyper)graphs, the problem was not investigated until recently, except in \cite{Courcelle06}
where it is proved that the models of any monadic second-order formula can be counted in polynomial time in graphs of bounded clique-width.  
Indeed, as far as we know the counting of minimal dominating sets is only considered in \cite{KanteLMNU13,GolovachHKKSV17,KanteU17}.
This problem is known to be polynomial on interval graphs and permutation graphs \cite{KanteLMNU13}. 
However, the systematic study of its computational complexity in graph classes is only considered in \cite{KanteU17}, where the authors proved the $\#$P-completeness in several graph classes and asked whether the following dichotomy conjecture is true. A $k$-sun is a graph obtained from a cycle of length $2k$ ($k\geq 3$) by adding edges to make the even-indexed vertices pairwise adjacent. A graph is \emph{chordal} if it does not contain cycles of length at least $4$ as induced subgraphs. 

\begin{conjecture}\label{conj:dichotomy-chordal}
  Let $\calC$ be a class of chordal graphs. If $\calC$ does not contain a $k$-sun as an induced subgraph, for $k\geq 4$, then one can count in polynomial time the number
  of minimal dominating sets of any graph in $\calC$. Otherwise, the problem is $\#$P-complete.
\end{conjecture}

This conjecture is motivated by the recursive structure of the chordal graphs without $k$-sun, for $k\leq 4$.

We make a first step towards a proof of the first statement of the conjecture and provide a polynomial time algorithm for computing the minimal dominating sets in \emph{strongly chordal graphs}, which are exactly chordal graphs without $k$-suns, for $k\geq 3$.

\begin{corollary}[Main Corollary] Let $G$ be a strongly chordal graph. One can count in polynomial time the number of minimal dominating sets of $G$.
\end{corollary}

\begin{proof} Let $G$ be a strongly chordal graph. It is well-known that the hypergraph $\cN[G]$ is $\beta$-acyclic
  \cite{BrandstadtLS99}. By Theorem \ref{thm:main}, one can count in polynomial time the minimal transversals of $\cN[G]$, which are exactly the minimal dominating sets
  of $G$. 
\end{proof}

\section{Conclusion}\label{sec:conclusion}

We proposed a polynomial time algorithm for counting the minimal transversals of any $\beta$-acyclic hypergraph, supporting Conjecture \ref{conj:dichotomy-chordal}, and it seems that the technique can be easily
adapted to consider (inclusion-wise) minimal $d$-dominating sets, which are dominating sets such that each vertex is dominated by at least $d$ vertices \cite{GolovachHKKSV17}. 

Besides resolving Conjecture \ref{conj:dichotomy-chordal}, there are two immediate questions that deserve to be considered. Firstly, can we count the minimal models of
any non-monotone $\beta$-acyclic formula in polynomial time? Secondly, for which graph classes and counting graph problems the techniques of this paper apply?

\bibliographystyle{elsarticle-num} 
\bibliography{biblio}

\end{document}